\documentclass[reqno,oneside]{amsart}

\usepackage{amsmath}
\usepackage{amsfonts}
\usepackage{amssymb}
\usepackage{enumitem}

\usepackage[normalem]{ulem}
\usepackage {ulem}

\usepackage[a4paper, hdivide={2.3cm,,3.8cm},vdivide={3.5cm,,3.8cm}]{geometry}

\usepackage{color}

\allowdisplaybreaks

\newcounter{maggicounter}
\newcommand{\maggia}[1]{\textcolor[rgb]{1,0,0}{*}\marginpar{\tiny{\arabic{maggicounter}: #1}}}
\renewcommand{\maggia}[1]{}

\newcommand{\R}{{\mathbb R}}

\newcommand{\N}{{\mathbb N}}
\newcommand{\Z}{{\mathbb Z}}

\newcommand{\wech}[1]{}

\newtheorem{theorem}{Theorem}[section]
\newtheorem{preremark}[theorem]{Remark}
\newenvironment{remark}{\begin{preremark}\normalfont}{\end{preremark}}
\newtheorem{lemma}[theorem]{Lemma}

\begin{document}

\begin{abstract}
An integration by parts formula is derived for the first order differential operator corresponding to the action of translations on the space of locally finite simple configurations of infinitely many points on $\mathbb{R}^d$. As reference measures, tempered grand canonical Gibbs measures are considered corresponding to a non-constant non-smooth intensity (one-body potential) and translation invariant potentials fulfilling the usual conditions. It is proven that such Gibbs measures fulfill the intuitive integration by parts formula if and only if the action of the translation is not broken for this particular measure. The latter is automatically fulfilled in the high temperature and low intensity regime. 
\end{abstract}

\title[An integration by parts formula]{A note on an integration by parts formula for the generators of uniform translations on configuration space}
\author{Florian Conrad, Tobias Kuna}

\date{\today}

\thanks{\textit{2010 AMS Mathematics Subject Classification}: 82B21, 82C22.}

\keywords{Gibbs measures, integration by parts formula, uniform translations.}

\address{\noindent Florian Conrad, Department of Mathematics, University of Kaiserslautern, P.O.Box 3049, 67653
Kaiserslautern, Germany and Mathematics Department, Bielefeld University, P.O.Box 100131, 33501 Bielefeld, Germany.
\newline {\rm \texttt{Email: fconrad@math.uni-bielefeld.de}}\newline
\indent Tobias Kuna, Department of Mathematics, The University of Reading, P.O.Box 220, Reading, Berkshire RG6 6AX/UK.
\newline {\rm \texttt{EMail: t.kuna@reading.ac.uk}}
}

\maketitle

\section{Introduction}
Infinite particle systems are the mathematical framework to describe complex systems of interacting individual objects or agents, like molecules in a liquid, stars in a galaxy, individuals in a population. The elementary states of the systems are all countable collections of points in $\mathbb{R}^d$ which have no accumulation point, i.e.~elements of $\Gamma$, the space of locally finite simple configurations in $\R^d$.

As it is typical for any infinite dimensional system, there does not exist a unique natural reference measure on $\Gamma$ singled out by its symmetry properties with respect to the action of a group.
In this paper, we consider reference measures from $\mathfrak G_{z\sigma}^{\beta\phi}$, the set of all (tempered grand canonical) Gibbs measures which describe a Hamiltonian system in the thermodynamic equilibrium. Gibbs measures are parametrized by the intensity $z\sigma$, pair-potential $\phi$ and inverse temperature $\beta$. 

Gibbs measures transfrom naturally under the standard action on $\Gamma$ of the group of the local diffeomorphisms on $\R^d$. Integration by parts is an infinitesimal manifestation of this natural transformation where a differential operator is the infinitesimal generator of the group action. In \cite{AKR98b} the authors build a calculus on configuration space based on the action of the local diffeomorphism group. Essentially by exploiting the local nature of the group action, they can derive an integration by parts formula for all tempered grand canoncial Gibbs measures. This integration by parts formula is one of the essential ingredients in the construction of a stochastic process using Dirichlet form techniques formally associated to an infinite system of stochastic differential equations driven by white noise, called stochastic gradient dynamics. This dynamics describes an infinite system of particles interacting via a pair potential $\phi$. The integration by parts formula yields the closability of the corresponding pre-Dirichlet form and is necessary for identifying the generator of the Dirichlet form on a set of local functions, i.e. smooth cylinder functions.

In this paper, instead, we consider the non-local action of the usual translation group shifting all points of the configuration simultaneously and in the same manner. We study the corresponding integration by parts formula for Gibbs measures. This integration by parts formula arises naturally, for example, in the construction of the environment process of a tagged particle, that is the movement of the  particles of a system seen from a tagged one. This corresponds heuristically to choosing a coordinate system in which the origin moves with the tagged particle. To give a rigorous meaning to the environment process associated to the aforementioned gradient dynamics, unfortunately, this intuition cannot be used and one has has to start the construction from scratch, for example, using Dirichlet form techniques, cf.~\cite{FG11}. Since the analytical objects (Dirichlet form, generator) used for this approach necessarily contain the information about a uniform component of the environment dynamics, the integration by parts formula for the generator of the translation group, which we derive here, is an important ingredient for this approach.

Let us list the challenges one has to overcome. First, the non-local nature of the translation group leads inevitably to boundary terms in a direct calculation if one uses a finite volume approximation. The derivation of the connection between generator and form in the aforementioned situation given in \cite{GP85} seems to neglect this difficulty and we are not aware of an easy fix that would work in great generality. We circumvent the problem by not using a local approximation and developing a different technique which has already been suggested in \cite{S09}. Second, even if the initial gradient dynamics corresponds to a translation invariant interaction, the associated tagged particle process has as invariant measure a Gibbs measure with non-translation invariant intensity, namely the interaction of the other particles with the tagged one. Moreover, it is physically reasonable that the mutual interaction contains a singular repulsion if particles get too near to each other. Thus we derive the integration by parts formula for general intensities $z\sigma_{\beta\psi}=ze^{-\beta\psi}\,dx$, where $\psi$ is integrable at infinity and $\psi(x)$ may grow at zero like $|x|^{-k}$ for some $k \in \mathbb{N}$, which covers physically relevant cases like e.g.~the Lennard-Jones potential. We include potentials which have singularities almost of an arbitrary nature and are not even everywhere weakly differentiable. Third, one can probably not expect that the integration by parts formula holds for every Gibbs measure, because e.g.~for constant intensities and translation invariant pair potentials there can exist non-translation invariant Gibbs measures due to a phase transition. For such a measure the integration by parts formula would look at least quite different from the intuitive one, if it exists. In the case of non-constant intensities, this problem will prevail or maybe even become worse. Hence, one can expect the integration by parts formula to hold in general only for elements from $\mathfrak G_{z\sigma_{\beta\psi}}^{\beta\phi}$ which are absolutely continuous with respect to a translation invariant element from $\mathfrak G_{zm}^{\beta\phi}$. Indeed, we can characterize that for given potential, inverse temperature and intensity $\sigma_{\beta\psi}$ these are the only measures which obey the intuitive integration by parts formula. 

As a first application, in \cite{CFGK11} the characterization is used to identify (in a natural way) the class of measures $\mu$ for which an invariance principle can be derived for the tagged particle dynamics (mentioned above) constructed from the Dirichlet form corresponding to $\mu$. It shows, for example, that using uniform motion for proving ergodicity of the environment process gives a more general result than an adaptation of the approach from \cite{AKR98b} (where a situation without uniform motion is considered).

While the proof of the integration by parts formula for these measures works under extremely weak assumptions on $\psi$, the characterization we can only show under slightly stronger conditions, which a priori seem not to be necessary. Nevertheless, they still allow nonintegrable singularities of $\nabla\psi$ at isolated points (or sets of sufficiently small dimension), covering in particular the usual non-hardcore pair interactions of statistical mechanics.

\section{Preliminaries and statement of the results}

By $\Gamma$ we denote the space of locally finite simple configurations in $\R^d$, i.e.~subsets $\gamma$ of $\R^d$ such that $\gamma\cap\Lambda$ is finite for all bounded $\Lambda\subset\R^d$. For a function $f: \R^d\to\R^k$, $k\in\N$, and $\gamma\in\Gamma$ we define
$$
\langle f,\gamma\rangle:=\sum_{x\in\gamma}f(x),
$$
if for each component $f_i$, $i=1,\cdots,k$, of $f=(f_1,\cdots,f_k)$ at least $\sum_{x\in\gamma} f_i^+(x)<\infty$ or $\sum_{x\in\gamma} f_i^-(x)<\infty$. Here and below we denote the positive and negative part of a real-valued function $g$ by $g^+$ and $g^-$, respectively. We denote by $\mathcal FC_b^\infty(C_0^\infty(\R^d),\Gamma)$ the set of all functions of the form $F=g_F(\langle f,\cdot\rangle)$ with $k\in\N$, $f=(f_1,\cdots,f_k)\in C_0^\infty(\R^d\to\R^k)$ and $g_F: \R^k\to \R$ infinitely often differentiable, bounded and such that all derivatives are bounded. The formal generator $\nabla_\gamma^\Gamma$ of the uniform translations on $\Gamma$ is given by $\nabla_\gamma^\Gamma F:=\sum_{j=1}^k\partial_j g_F(\langle f,\cdot\rangle)\,\langle \nabla f_j,\cdot\rangle$, $F\in \mathcal FC_b^\infty(C_0^\infty(\R^d),\Gamma)$ as above.

$\Gamma$ is equipped with the $\sigma$-field $\mathcal B$ generated by the mappings $\gamma\mapsto \langle 1_\Lambda,\gamma\rangle$, $\Lambda\subset\R^d$ measurable and bounded. The objects under consideration are tempered grand canonical Gibbs measures on $(\Gamma,\mathcal B)$ for superstable and lower regular pair potentials $\phi$ with intensity measure $\sigma$ having a bounded density w.r.t.~Lebesgue measure, activity $z>0$ and inverse temperature $\beta>0$. Gibbs measures can be defined in several different but completely equivalent ways, see e.g.~\cite{KK03}. The one used here is recalled in the appendix. The set of all tempered grand canonical Gibbs measures for $z$, $\beta$, $\phi$, $\sigma$ is denoted by $\mathfrak G_{z\sigma}^{\beta\phi}$.

By $m$ we denote Lebesgue measure on $\R^d$. Note that if $C<\infty$, and $\psi:\R^d\to[-C,\infty]$ is measurable, then $e^{-\beta\psi}$ is weakly differentiable iff for all $n\in\N$ the function $\psi\wedge n$ is weakly differentiable and $\sup_{n\in\N}\Vert \nabla(\psi\wedge n) e^{-\beta\psi}\Vert_{L^1(K;m)}<\infty$ for all compact $K\subset\R^d$. In this case we define $\nabla\psi:=-(1_{[-C,\infty)}\circ\psi) e^{\beta\psi} \beta^{-1}\nabla e^{-\beta\psi}$ and observe that $\nabla(\psi\wedge n)=(1_{[-C,n]}\circ\psi)\nabla\psi$, $n\in\N$. Let us state the first main result of this note.

\begin{theorem}\label{thm}
Let $\phi:\R^d\to\R\cup\{\infty\}$ be a (measurable, even,) superstable, lower regular potential and let $\beta>0$ and $z>0$. Moreover, let $\psi: \R^d\to\R\cup\{\infty\}$ be measurable, bounded from below and such that $(1-e^{-\beta\psi})\in L^1(\R^d;m)$. Define the measure $\sigma_{\beta\psi}:=e^{-\beta\psi}m$. Then the following assertions hold:
\begin{enumerate}
\item For any $\mu_{zm}^{\beta\phi}\in \mathfrak G_{zm}^{\beta\phi}$ a measure $\mu_{{z\sigma_{\beta\psi}}}^{\beta\phi}\in\mathfrak G_{z\sigma_{\beta\psi}}^{\beta\phi}$ is given by
\begin{equation}\label{eqn:einszueinsabbildung}
\frac{d\mu_{z\sigma_{\beta\psi}}^{\beta\phi}}{d\mu_{zm}^{\beta\phi}}=\Xi_{\psi}^{-1} e^{-\beta\langle\psi,\cdot\rangle},
\end{equation}
where $\Xi_{\psi}:=\int_\Gamma e^{-\beta\langle\psi,\gamma\rangle} d\mu_{zm}^{\beta\phi}(\gamma)$. If $\psi$ is Lebesgue-a.e.~finite, in a similar manner for any element from $\mathfrak G_{z\sigma_{\beta\psi}}^{\beta\phi}$ an element from $\mathfrak G_{zm}^{\beta\phi}$ is obtained and \eqref{eqn:einszueinsabbildung} gives a bijection between $\mathfrak G_{zm}^{\beta\phi}$ and $\mathfrak G_{z\sigma_{\beta\psi}}^{\beta\phi}$.
\item If, in addition to the assumptions preceding (i), $e^{-\beta\psi}$ is weakly differentiable, $\nabla\psi$ (defined as above) is integrable w.r.t.~$\sigma_{\beta\psi}$, and $\mu_{zm}^{\beta\phi}\in\mathfrak G_{zm}^{\beta\phi}$ is translation invariant, then for $\mu_{z\sigma_{\beta\psi}}^{\beta\phi}$ as in (i) we obtain the following integration by parts formula: For every $F\in \mathcal FC_b^\infty(C_0^\infty(\R^d),\Gamma)$ it holds
\begin{equation}\label{eqn:ibp}
\int_\Gamma \nabla_\gamma^\Gamma F\,d\mu_{z\sigma_{\beta\psi}}^{\beta\phi}=\beta\int_\Gamma F\langle\nabla\psi,\cdot\rangle\,d\mu_{z\sigma_{\beta\psi}}^{\beta\phi}.
\end{equation}
\end{enumerate}
\end{theorem}

\begin{remark}\label{rem}
Let $\phi$ be a potential fulfilling the assumptions of Theorem \ref{thm}. 
\begin{enumerate}
\item It is well-known that if $\mu\in\mathfrak G_{z\sigma}^{\beta\phi}$ for $\sigma=m$, $\beta>0$ and $z> 0$, then $\mu$ fulfills the so-called Ruelle bound (see \cite[Eq.~(5.28)]{Ru70} for the meaning of this statement). By analyzing the proof of the last part of Corollary 5.3 in \cite{Ru70} it is not difficult to see that the Ruelle bound extends to all $\mu\in\mathfrak G_{z\sigma}^{\beta\phi}$ in the case when $\sigma$ has a bounded density w.r.t.~Lebesgue measure, i.e.~in particular when $\sigma=\sigma_{\beta\psi}$ as defined in Theorem \ref{thm}. 
\item If additionally $\int_{\R^d}\vert e^{-\beta\phi}-1\vert\,dm<\infty$, it is known (see \cite{Ru70}) that $\mathfrak G_{zm}^{\beta\phi}\neq\emptyset$ and there also exists a translation invariant element of $\mathfrak G_{zm}^{\beta\phi}$. Thus Theorem \ref{thm}(ii) implies the existence of an element of $\mathfrak G_{z\sigma_{\beta\psi}}^{\beta\phi}$ fulfilling \eqref{eqn:ibp}.
\item Although in the situation from (ii) one could derive from Theorem \ref{thm}(i) also that $\mathfrak G_{z\sigma_{\beta\psi}}^{\beta\phi}\neq \emptyset$, it is more natural to derive this existence result from the construction in \cite{Ru70} or \cite{KK03}, since \cite[Proposition 2.6]{Ru70} is easily seen to extend to the case of intensity measures $\sigma$ having bounded density w.r.t.~Lebesgue measure.
\end{enumerate}
\end{remark}

Let $\mu_{zm}^{\beta\phi}$ and $\mu_{z\sigma_{\beta\psi}}^{\beta\phi}$ be as in Theorem \ref{thm}(i) and let $\psi$ 	fulfill the additional assumptions from Theorem \ref{thm}(ii). The question arises whether \eqref{eqn:ibp} is not only implied by but even characterizes translation invariance of $\mu_{zm}^{\beta\phi}$. This is a natural conjecture, because translation invariance of $\mu_{zm}^{\beta\phi}$ is equivalent to $\mu_{z\sigma_{\beta\psi}}^{\beta\phi}$ being quasi-invariant w.r.t.~the translations $\theta_v: \gamma\mapsto \gamma+v$, $v\in\R^d$, with density $\frac{d\mu_{z\sigma_{\beta\psi}}^{\beta\phi}\circ \theta_v^{-1}}{d\mu_{z\sigma_{\beta\psi}}^{\beta\phi}}(\gamma)=e^{-\beta\langle \psi,\gamma-v\rangle+\beta\langle \psi,\gamma\rangle}$ and because \eqref{eqn:ibp} is just the differential version of the latter statement. Here we define $\gamma+v:=\{x+v\,|\,x\in\gamma\}$ for $\gamma\in\Gamma$, $v\in\R^d$. In the next theorem we verify this conjecture under some more conditions on $\psi$; the difficulties to treat the general case are explained in Remark \ref{rem:notgeneral} below.

\begin{theorem}\label{thm2}
Assume that $d\geq 2$. Let $\phi$, $\psi$ be as in Theorem \ref{thm}(ii) and assume additionally that $\psi$ is weakly differentiable in $\R^d\setminus\{0\}$ and $\nabla\psi\in L^1(\R^d\setminus B_1(0))$, where $B_1(0)$ denotes the ball around $0$ with radius $1$. Let $\mu_{zm}^{\beta\phi}$, $\mu_{z\sigma_{\beta\psi}}^{\beta\phi}$ be as in Theorem \ref{thm}(i). Then $\mu_{zm}^{\beta\phi}$ is translation invariant iff $\mu_{z\sigma_{\beta\psi}}^{\beta\phi}$ fulfills \eqref{eqn:ibp}.
\end{theorem}

\begin{remark}
\begin{enumerate}
\item For the applications mentioned in the introduction, the additional restrictions in the previous theorem are irrelevant. They still allow the usual non-hardcore pair potentials from statistical physics, e.g.~the Lennard-Jones potential. These potentials are usually bounded outisde any neighborhood of the origin.
\item For $d=1$ the conclusion of Theorem \ref{thm2} is in most cases trivial: In this case usually $\mathfrak G_{zm}^{\beta\phi}$ consists for all $z,\beta>0$ only of one element, which is automatically translation invariant. This was shown in \cite{Pap87} for superstable potentials $\phi$ which are bounded on the complement of any neighborhood of $0$ and have the property that there exists a decreasing function $\varphi$ such that $\vert\phi(x)\vert\leq \varphi(\vert x\vert)$, $\vert x\vert \geq R>0$, and $\int_R^\infty \varphi(x)\,dx<\infty$. In particular, these conditions cover the usual type of potentials from statistical mechanics.
\item The proof of Theorem \ref{thm2} given below extends to the case when one only assumes that $\psi$ is (as in Theorem 2.1(ii) and) weakly differentiable in $\R^d\setminus K$ for some compact $K\subset\R^d$ having Hausdorff dimension strictly less than $d-1$ (instead of choosing $K=\{0\}$) and $\nabla\phi\in L^1(\R^d\setminus B_R(0);m)$ with $R$ large enough such that $K\subset B_R(0)$.
\item Theorem \ref{thm2} shows that the one-to-one correspondence from Theorem \ref{thm}(i) extends also to a one-to-one correspondence between the set of translation invariant elements from $\mathfrak G_{zm}^{\beta\phi}$ and the set of elements from $\mathfrak G_{z\sigma_{\beta\psi}}^{\beta\phi}$ fulfilling \eqref{eqn:ibp}. As one easily verifies, the correspondence preserves the structure of these sets in the sense that extremal elements in the former set (pure phases) correspond to extremal elements in the latter set.
\end{enumerate}
\end{remark}

\section{Proofs}

For a measure $\sigma$ on $\R^d$ and a measurable function $f: \R^d\to\R$ we define $C_{f,\sigma}:=\int_{\R^d}\vert e^f-1\vert\,d\sigma$. We need the following lemma, which is essentially contained in \cite{KK02}, \cite{KK03}. 
\begin{lemma}\label{lem}
Let $\sigma$ be a measure on $\R^d$ having a bounded density w.r.t.~$m$ and let $\mu_{z\sigma}^{\beta\phi}\in\mathfrak G_{z\sigma}^{\beta\phi}$.
\begin{enumerate}
\item Let $N\subset\R^d$ with $\sigma(N)=0$. Then $\mu_{z\sigma}^{\beta\phi}$-a.s.~it holds $\gamma\subset\R^d\setminus N$.
\item Let $f\in L^1(\R^d;\sigma)$. Then $\langle f,\cdot\rangle=\sum_{x\in\cdot}f(x)$ converges absolutely $\mu_{z\sigma}^{\beta\phi}$-a.s. Moreover, $\Vert \langle f,\cdot\rangle\Vert_{L^1(\Gamma;\mu_{z\sigma}^{\beta\phi})}\leq \xi_\sigma\Vert f\Vert_{L^1(\R^d;\sigma)}$ for some $\xi_\sigma<\infty$ which is independent of $f$.
\item Let $\vert f\vert\wedge 1\in L^1(\R^d;\sigma)$ and $f$ be finite $\sigma$-a.e., then $\langle f,\cdot\rangle=\sum_{x\in\cdot}f(x)$ converges absolutely $\mu_{z\sigma}^{\beta\phi}$-a.s. If only $f^+$ is $\sigma$-a.e.~finite then for $\mu_{z\sigma}^{\beta\phi}$-a.e.~$\gamma\in\Gamma$ there exists a finite configuration $\eta\subset\gamma$ such that $f(x)>-\infty$ for all $x\in \gamma\setminus\eta$ and $\sum_{x\in\gamma\setminus \eta}f(x)$ converges absolutely. Moreover one can choose $\eta=\emptyset$ with nonzero $\mu_{z\sigma}^{\beta\phi}$-probability.
\item Let $f:\R^d\to\R\cup\{-\infty\}\cup\{\infty\}$ be measurable and such that $C_{f,\sigma}<\infty$. Then $e^{\langle f,\cdot\rangle}:\Gamma\to [0,\infty)$ is well-defined and integrable w.r.t.~$\mu_{z\sigma}^{\beta\phi}$. $e^{\langle f,\cdot\rangle}$ is positive with positive $\mu_{z\sigma}^{\beta\phi}$-probability, and if $f$ is $\sigma$-a.e.~finite, then $e^{\langle f,\cdot\rangle}$ is $\mu_{z\sigma}^{\beta\phi}$-a.s.~positive.
\end{enumerate}
\end{lemma}
\begin{proof}
For (i) and (ii) see e.g.~\cite[Theorem 4.1]{KK02} (note that $\mu_{z\sigma}^{\beta\phi}$ fulfills a Ruelle bound, see Remark \ref{rem}(i)). For proving (iii) let $A:=\{x\in\R^d\,|\,\vert f(x)\vert\geq 1\}$. Then $A$ has finite $\sigma$-measure and hence by (ii) it holds $\langle 1_A,\cdot\rangle\in L^1(\Gamma;\mu_{z\sigma}^{\beta\phi})$, and thus $\sharp(\gamma\cap A)<\infty$ $\mu_{z\sigma}^{\beta\phi}$-a.s.~and by the definition of grand canonical Gibbs measures $\sharp(\gamma\cap A)=0$ with positive $\mu_{z\sigma}^{\beta\phi}$-probability. Moreover, since $1_{\R^d\setminus A} f\in L^1(\R^d;\sigma)$, (ii) also implies that $\sum_{x\in \gamma\setminus A}f(x)$ converges absolutely $\mu_{z\sigma}^{\beta\phi}$-a.s. Together with (i) the assertions in (iii) follow. Since $\vert f\vert\wedge 1\leq e^1\vert e^f-1\vert$ and since moreover $C_{f,\sigma}<\infty$ implies that $f^+$ is $\sigma$-a.e.~finite, (iv) is a consequence of (iii), with the exception of the integrability statement, which is seen as follows: We may w.l.o.g.~assume that $f\geq 0$. Then $e^{\langle 1_{[-n,n]^d} f,\cdot\rangle}\uparrow e^{\langle f,\cdot\rangle}$ as $n\to \infty$. Since $C_{f^+,\sigma}\leq C_{f,\sigma}<\infty$, e.g.~the proof of \cite[Proposition 5.1]{KK03} implies that $\sup_{n\in \N}\Vert e^{\langle 1_{[-n,n]^d}f,\cdot\rangle}\Vert_{L^1(\Gamma;\mu_{z\sigma}^{\beta\phi})}<\infty$. Using the monotone convergence theorem, we therefore obtain $e^{\langle f,\cdot\rangle}\in L^1(\Gamma;\mu_{z\sigma}^{\beta\phi})$.
\end{proof}

\begin{proof}[Proof of Theorem \ref{thm}]
Part (i): We start with $\mu_{zm}^{\beta\phi}\in \mathfrak G_{zm}^{\beta\phi}$. It follows from Lemma \ref{lem}(iv) that $e^{-\beta\langle\psi,\cdot\rangle}$ is not $\mu_{zm}^{\beta\phi}$-a.s.~equal to zero and integrable w.r.t.~$\mu_{zm}^{\beta\phi}$ and hence $0< \Xi_{\psi}<\infty$. The Ruelle equation (R) as given in the appendix, putting $\mu=\mu_{zm}^{\beta\phi}$ and $\sigma=m$, implies that $e^{-\beta\langle\psi,\cdot\rangle}\mu_{zm}^{\beta\phi}$ fulfills (R) with $\sigma=\sigma_{\beta\psi}$; that can be seen by replacing $F$ by $Fe^{-\beta\langle\psi,\cdot\rangle}$ in (R). Hence, we proved that $\Xi_\psi^{-1} e^{-\beta\langle\psi,\cdot\rangle}\mu_{zm}^{\beta\phi}$ is a (tempered) grand canonical Gibbs measure for $\phi$ with intensity measure $\sigma_{\beta\psi}$. Conversely, since $C_{\beta\psi,\sigma_{\beta\psi}}=C_{-\beta\psi,m}<\infty$, when starting with $\mu_{z\sigma_{\beta\psi}}^{\beta\phi}\in \mathfrak G_{z\sigma_{\beta\psi}}^{\beta\phi}$, we obtain $e^{\beta\langle\psi,\cdot\rangle}\in L^1(\Gamma;\mu_{z\sigma_{\beta\psi}}^{\beta\phi})$ by Lemma \ref{lem}(iv). If $\psi$ is finite Lebesgue-a.e., it follows that $e^{\beta\psi}\sigma_{\beta\psi}=m$. Using this and the Ruelle equation (R), we can show as above that the normalized $e^{\beta\langle\psi,\cdot\rangle}\mu_{z\sigma_{\beta\psi}}^{\beta\phi}$ is in $\mathfrak G_{zm}^{\beta\phi}$. The last assertion of Theorem \ref{thm}(i) follows from the last assertion in Lemma \ref{lem}(iv).

Part (ii): Let $\mu_{zm}^{\beta\phi}\in\mathfrak G_{zm}^{\beta\phi}$ be translation invariant. We first consider the case $\psi=0$, but for a more general class of $F$. As before we assume that $F$ is of the form $F=g_F(\langle f,\cdot\rangle)$, but for $g_F$ we only require (in addition to smoothness) that $g_F$ and $\nabla g_F$ are exponentially bounded (i.e.~bounded in absolute value by $Ce^{a\vert\cdot\vert}$ for some $C<\infty$ and $a\in\R$, where $\vert \cdot\vert$ is Euclidean norm). For such functions, $v\in\R^d$, $\gamma\in\Gamma$ and $t\in [0,1]$ it holds $\frac{d}{dt}F(\gamma+tv)=v\nabla_\gamma^\Gamma F(\gamma+tv)$, and thus by the mean value theorem
\begin{equation}\label{eqn:dry}
\vert F(\gamma+vt)-F(\gamma)\vert\leq t \sup_{t'\in [0,1]}\vert v\nabla_\gamma^\Gamma F(\gamma+t'v)\vert\leq t\tilde C e^{\tilde a\langle 1_\Lambda,\gamma\rangle}
\end{equation}
for some $\tilde C<\infty$ and $\tilde a\in\R$, both not depending on $\gamma$, and for some open bounded $\Lambda\subset\R^d$ containing all points having distance less than $\vert v\vert$ from the support of $f$. By Lemma \ref{lem}(iv), the right-hand side of \eqref{eqn:dry} is in $L^1(\Gamma;\mu_{zm}^{\beta\phi})$, hence by Lebesgue's dominated convergence theorem $\int_{\Gamma} v\nabla_\gamma^\Gamma F\,d\mu_{zm}^{\beta\phi}=\lim_{t\to 0} \frac{1}{t} \int_\Gamma F(\cdot+vt)-F\,d\mu_{zm}^{\beta\phi}$, and the latter is equal to $0$ by the assumed translation invariance of $\mu_{zm}^{\beta\phi}$. Thus \eqref{eqn:ibp} holds for $\psi=0$. Replacing $F$ by $F e^{-\beta\langle\psi,\cdot\rangle}$, we obtain directly \eqref{eqn:ibp} also for $\psi\in C_0^\infty(\R^d)$, using that $\mu_{z\sigma_{\beta\psi}}^{\beta\phi}=\frac{1}{\Xi_\psi} e^{-\beta\langle \psi,\cdot\rangle}\mu_{zm}^{\beta\phi}$.

We derive the general result extending this by three approximation arguments. (From now on, we restrict again to $F$ as in the assertion.) First, let $\psi\in H^{1,1}(\R^d)$ be bounded and compactly supported. Using convolutions with a Dirac sequence, we obtain a sequence $(\psi_n)_{n\in\N}\subset C_0^\infty(\R^d)$ having the following properties:
\begin{enumerate}
\item[(a)] $\psi_n\to\psi$ in $H^{1,1}(\R^d)$ as $n\to\infty$.
\item[(b)] There exists $0\leq \psi_0\in L^1(\R^d;m)\cap L^\infty(\R^d;m)$ such that $\vert \psi_n\vert \leq\psi_0$.
\end{enumerate}
For later use we emphasize that we extend \eqref{eqn:ibp} to $\psi$ using only (a), (b) and the fact that \eqref{eqn:ibp} holds for all $\psi_n$. By dropping to a subsequence we may assume that $\psi_n\to \psi$ holds pointwise Lebesgue-a.e. Lemma \ref{lem}(ii) implies that $\langle \psi_0,\cdot\rangle<\infty$ holds $\mu_{zm}^{\beta\phi}$-a.e. By (b), Lebesgue's dominated convergence theorem and Lemma \ref{lem}(i) we conclude that $\langle \psi_n,\cdot\rangle\to \langle \psi,\cdot\rangle$ as $n\to\infty$ pointwise $\mu_{zm}^{\beta\phi}$-a.e. Using Lebesgue's theorem, (b) and integrability of $\vert \nabla_\gamma^\Gamma F\vert e^{\beta\langle\psi_0,\cdot\rangle}$ w.r.t.~$\mu_{zm}^{\beta\phi}$ (which follows e.g.~using Lemma \ref{lem}(iv), since $\vert \nabla_\gamma^\Gamma F\vert$ can be estimated by $C\langle 1_\Lambda,\cdot\rangle\leq C e^{\langle 1_\Lambda,\cdot\rangle}$ for some $C<\infty$ and some open bounded $\Lambda\subset\R^d$) we thus find that for $F\in\mathcal FC_b^\infty(C_0^\infty(\R^d),\Gamma)$ it holds
\begin{equation}\label{eq1}
\int_\Gamma \nabla_\gamma^\Gamma F\,e^{-\beta\langle \psi_n,\cdot\rangle}\,d\mu_{zm}^{\beta\phi}\to \int_\Gamma \nabla_\gamma^\Gamma F\,e^{-\beta\langle \psi,\cdot\rangle}\,d\mu_{zm}^{\beta\phi}
\end{equation}
as $n\to\infty$, i.e.~we have convergence of the left-hand side of \eqref{eqn:ibp}. In order to prove convergence of the right-hand side, we show that
\begin{equation}\label{eq2a}
\int_\Gamma \left\vert F \langle \nabla\psi-\nabla\psi_n,\cdot\rangle\right\vert\,e^{-\beta\langle \psi,\cdot\rangle}\,d\mu_{zm}^{\beta\phi}\to 0
\end{equation}
and
\begin{equation}\label{eq2b}
\int_\Gamma \left\vert F\,\langle \nabla\psi_n,\cdot\rangle\,\left(e^{-\beta\langle \psi,\cdot\rangle}-e^{-\beta\langle \psi_n,\cdot\rangle}\right)\right\vert\,d\mu_{zm}^{\beta\phi}\to 0
\end{equation}
as $n\to\infty$. By (a) we have $\nabla\psi_n\to\nabla\psi$ in $L^1(\R^d;m)=L^1(\R^d;\sigma_{-\beta\psi_0})$, hence Lemma \ref{lem}(ii) implies that the sequence $(\langle\nabla\psi_n,\cdot\rangle)_{n\in\N}$ converges to $\langle\nabla\psi,\cdot\rangle$ in $L^1(\Gamma;\mu_{z\sigma_{-\beta\psi_0}}^{\beta\phi})$. This implies \eqref{eq2a}, since the left-hand side of \eqref{eq2a} can be estimated by $\Xi_{-\psi_0} \Vert F\Vert_\infty\int_\Gamma \vert \langle \nabla\psi-\nabla\psi_n,\cdot\rangle \vert d\mu_{z\sigma_{-\beta\psi_0}}^{\beta\phi}$. To prove \eqref{eq2b}, we use that convergence of $(\langle\nabla\psi_n,\cdot\rangle)_{n\in\N}$ in $L^1(\Gamma;\mu_{z\sigma_{-\beta\psi_0}}^{\beta\phi})$ implies uniform integrability of this sequence w.r.t.~$\mu_{z\sigma_{-\beta\psi_0}}^{\beta\phi}$. For any $a\in\R$ we have
\begin{multline*}
\int_\Gamma \left\vert F\,\langle \nabla\psi_n,\cdot\rangle\,\left(e^{-\beta\langle \psi,\cdot\rangle}-e^{-\beta\langle \psi_n,\cdot\rangle}\right)\right\vert\,d\mu_{zm}^{\beta\phi}\\
\leq 2\Xi_{-\psi_0}\Vert F\Vert_\infty\,\int_{\vert\langle \nabla\psi_n,\cdot\rangle\vert\geq a}\vert \langle\nabla\psi_n,\cdot\rangle\vert\,\,d\mu_{z\sigma_{-\beta\psi_0}}^{\beta\phi}
+a\Vert F\Vert_\infty\int_\Gamma \left\vert e^{-\beta\langle\psi_n,\cdot\rangle}-e^{-\beta\langle \psi,\cdot\rangle}\right\vert d\mu_{zm}^{\beta\phi}.
\end{multline*}
The first summand on the right-hand side can be made arbitrarily small uniformly in $n$ by choosing $a$ large, the second converges to $0$ as $n\to\infty$ for any fixed $a\in\R$. From this \eqref{eq2b} follows. Hence, \eqref{eqn:ibp} is verified for bounded, compactly supported $\psi\in H^{1,1}(\R^d)$.

We now give the second approximation argument in order to treat the case when {$\psi\in H^{1,1}(\R^d)$} is bounded, but not necessarily compactly supported: Choose a sequence $(\chi_n)_{n\in\N}\subset C_0^\infty(\R^d)$ such that $1_{[-n,n]^d}\leq \chi_n\leq 1_{[-2n,2n]^d}$ and $\Vert \nabla\chi_n\Vert_\infty\to 0$ as $n\to\infty$, and define $\psi_n:=\chi_n\,\psi$. By the above considerations we know that \eqref{eqn:ibp} holds for all $\psi_n$, $n\in\N$. In order to extend \eqref{eqn:ibp} to $\psi$, we can apply precisely the same arguments as above, since (a) and (b) are again valid with $\psi_0=\vert \psi\vert\in L^1(\R^d;m)\cap L^\infty(\R^d;m)$.

The following (third) approximation argument extends \eqref{eqn:ibp} to general $\psi$ as in the assertion: Setting $\psi_n:=\psi\wedge n$, $n\in\N$, we again obtain an approximating sequence of functions fulfilling \eqref{eqn:ibp}. In order to prove \eqref{eq1} we use the following arguments which are a slight modification of the above ones: Since for any $n\in\N$ it holds $\psi_n^-=\psi^-$ we have $\langle \psi_n^-,\cdot\rangle=\langle\psi^-,\cdot\rangle$, and this is finite $\mu_{zm}^{\beta\phi}$-a.s.~by Lemma \ref{lem}(ii). Moreover, it holds $\langle \psi^+_n,\gamma\rangle\to \langle \psi^+,\gamma\rangle\in [0,\infty]$ as $n\to\infty$ for any $\gamma\in\Gamma$ by the monotone convergence theorem. Thus we obtain $\langle \psi_n,\cdot\rangle\to\langle \psi,\cdot\rangle$ as $n\to\infty$ pointwise $\mu_{zm}^{\beta\phi}$-a.s. Since $\vert\nabla_\gamma^\Gamma F\vert e^{-\beta\langle \psi_n,\cdot\rangle}\leq \vert \nabla_\gamma^\Gamma F\vert e^{\beta\langle \psi^-,\cdot\rangle}\in L^1(\Gamma;\mu_{zm}^{\beta\phi})$, \eqref{eq1} follows by Lebesgue's dominated convergence theorem. Moreover, since
$$
\int_\Gamma \left\vert F\,\langle \nabla\psi-\nabla\psi_n,\cdot\rangle\right\vert\,e^{-\beta\langle \psi,\cdot\rangle}\,d\mu_{zm}^{\beta\phi}\leq \Xi_{\psi}\cdot\Vert F\Vert_\infty \Vert \langle\nabla\psi-\nabla\psi_n,\cdot\rangle\Vert_{L^1(\Gamma;\mu_{z\sigma_{\beta\psi}}^{\beta\phi})},
$$
by Lemma \ref{lem}(ii) and the fact that $\Vert \nabla\psi_n-\nabla\psi\Vert_{L^1(\R^d;\sigma_{\beta\psi})}=\int_{\R^d} (1_{(n,\infty)}\circ\psi) \vert \nabla\psi \vert e^{-\beta\psi}\,dx\to 0$ as $n\to\infty$, we obtain \eqref{eq2a}. In order to show \eqref{eq2b} we define the measure $\mathcal B(\R^d\times\Gamma)\ni A\mapsto \mu_{*}(A):=\int_{\Gamma} \sum_{x\in\gamma} 1_A(x,\gamma)\,d\mu_{zm}^{\beta\phi}(\gamma)$ on the Borel $\sigma$-field $\mathcal B(\R^d\times\Gamma)$ of $\R^d\times\Gamma$. It holds
$$
\int_{\Gamma} \left \vert F\langle\nabla\psi_n,\cdot\rangle\left(e^{-\beta\langle \psi,\cdot\rangle}-e^{-\beta\langle \psi_n,\cdot\rangle}\right)\right\vert\,d\mu_{zm}^{\beta\phi}\leq \Vert F\Vert_\infty \int_{\R^d\times\Gamma} \Theta_n(x,\gamma)\,d\mu_*(x,\gamma),
$$
where $\Theta_n(x,\gamma):=\vert\nabla\psi_n(x)\vert \left\vert e^{-\beta\langle\psi,\gamma\rangle}-e^{-\beta\langle\psi_n,\gamma\rangle}\right\vert$ for $(x,\gamma)\in\R^d\times\Gamma$ such that $\langle \psi^-,\gamma\rangle<\infty$. Note that $\Theta_n$ is $\mu_*$-a.e.~defined. We have to prove convergence of $\Theta_n$ to $0$ in $L^1(\R^d\times\Gamma;\mu_*)$ as $n\to\infty$. To this end, we first note that for any $n\in\N$, $\gamma\in\Gamma$ (s.t.~$\langle \psi^-,\gamma\rangle<\infty$) and $x\in\gamma$ it holds
\begin{align*}
\Theta_n(x,\gamma)&=1_{(-\infty,n]}(\psi(x)) \vert\nabla\psi(x)\vert \left\vert e^{-\beta\langle\psi,\gamma\rangle}-e^{-\beta\langle\psi_n,\gamma\rangle}\right\vert\\
&=1_{(-\infty,n]}(\psi(x)) \vert\nabla\psi(x)\vert e^{-\beta\psi(x)}\left\vert e^{-\beta\langle\psi,\gamma\setminus \{x\}\rangle}-e^{-\beta\langle\psi_n,\gamma\setminus \{x\}\rangle}\right\vert
\end{align*}
The right-hand side converges to $0$ as $n\to\infty$. This shows that the sequence $(\Theta_n)_{n\in\N}$ converges pointwise to $0$ $\mu_{*}$-a.e. In order to obtain convergence to $0$ in $L^1(\R^d\times\Gamma;\mu_*)$ from Lebesgue's dominated convergence theorem, we note that for $\gamma\in\Gamma$ (s.t.~$\langle \psi^-,\gamma\rangle<\infty$) and $x\in\gamma$ the above equality implies
\begin{align*}
\Theta_n(x,\gamma)&\leq \vert\nabla\psi(x)\vert e^{-\beta\psi(x)} e^{\beta\langle \psi^-,\gamma\setminus\{x\}\rangle}\\
&= \vert\nabla\psi(x)\vert e^{-\beta\psi^+(x)} e^{\beta\langle\psi^-,\gamma\rangle}=:\Theta_0(x,\gamma)
\end{align*}
and
\begin{align*}
\int_{\R^d\times\Gamma} \Theta_0\,d\mu_*&=\int_{\Gamma}\langle\vert \nabla\psi\vert e^{-\beta\psi^+},\cdot\rangle e^{\langle\beta\psi^-,\cdot\rangle}\,d\mu_{zm}^{\beta\phi}=\Xi_{-\psi^-}\int_{\Gamma}\langle \vert \nabla\psi\vert e^{-\beta\psi^+},\cdot\rangle\,d\mu_{z\sigma_{-\beta\psi^-}}^{\beta\phi}\\&\leq \xi_{\sigma_{-\beta\psi^-}}\Xi_{-\psi^-} \int_{\R^d}\vert\nabla\psi\vert e^{-\beta\psi^+}e^{\beta\psi^-}\,dm
\end{align*}
by Lemma \ref{lem}(ii), and the right-hand side is finite by assumption.

\end{proof}

We now turn to the proof of Theorem \ref{thm2}. A first step to the proof is contained in the following lemma.

\begin{lemma}\label{lem2}
Let $\phi$, $\psi$ be as in Theorem \ref{thm2}. Let $\mu_{zm}^{\beta\phi}$, $\mu_{z\sigma_{\beta\psi}}^{\beta\phi}$ be as in Theorem \ref{thm}(i) and assume that \eqref{eqn:ibp} holds. Let $\varphi: \R^d\to\R$ be weakly differentiable and such that $\varphi\leq \psi$, $\varphi^-\in L^1(\R^d;m)$, $\nabla\varphi\in L^1(\R^d;m)$ and $\sup_{y\in A}\varphi(y)<\infty$ for a neighborhood $A$ of $0$. Then
\begin{equation}\label{eqn:pretinv}
\int_{\Gamma} \nabla_\gamma^\Gamma F \cdot e^{\beta\langle\varphi,\cdot\rangle}d\mu_{\beta\sigma_{\beta\psi}}^{\beta\phi}=\beta\int_{\Gamma} \langle \nabla\psi-\nabla\varphi,\cdot\rangle e^{\beta\langle\varphi,\cdot\rangle}d\mu_{\beta\sigma_{\beta\psi}}^{\beta\phi}.
\end{equation}
\end{lemma}
\begin{proof}
Observe that for any $K\in\N$ it holds $\varphi\wedge K\in H^{1,1}(\R^d)$. Choose a sequence $(\varphi_n)_{n\in\N}\subset C_0^\infty(\R^d)$ such that $\varphi_n\to \varphi\wedge K$ and $\nabla\varphi_n\to\nabla(\varphi\wedge K)$ in $L^1(\R^d)$. We apply \eqref{eqn:ibp} to obtain
$$
\int_{\Gamma} \nabla_\gamma^\Gamma F \cdot h(\langle\varphi_n,\cdot\rangle))\,d\mu_{z\sigma_{\beta\psi}}^{\beta\phi}=-\int_{\Gamma} F\cdot h'(\langle\varphi_n,\cdot\rangle))\langle \nabla\varphi_n,\cdot\rangle\,d\mu_{z\sigma_{\beta\psi}}^{\beta\phi}+\beta\int_{\Gamma} F \cdot h(\langle\varphi_n,\cdot\rangle))\langle \nabla\psi,\cdot\rangle\,d\mu_{z\sigma_{\beta\psi}}^{\beta\phi}
$$
for $h\in C_b^\infty(\R)$ and consider the limit as $n\to\infty$. We have $\langle\varphi_n,\cdot\rangle\to \langle\varphi\wedge K,\cdot\rangle$ in $L^1(\Gamma;\mu_{z\sigma_{\beta\psi}}^{\beta\phi})$ by Lemma \ref{lem}(ii), and dropping to a subsequence we may w.l.o.g.~assume that this holds also pointwise $\mu_{z\sigma_{\beta\psi}}^{\beta\phi}$-a.s. Thus $h(\langle\varphi_n,\cdot\rangle)\to h(\langle\varphi\wedge K,\cdot\rangle)$ and $h'(\langle\varphi_n,\cdot\rangle)\to h'(\langle\varphi\wedge K,\cdot\rangle)$ pointwise $\mu_{z\sigma_{\beta\psi}}^{\beta\phi}$-a.s.~and by Lebesgue's theorem also in weak-$*$ sense in $L^\infty(\Gamma;\mu_{z\sigma_{\beta\psi}}^{\beta\phi})$. Together with the convergence $\langle\nabla\varphi_n,\cdot\rangle\to \langle \nabla(\varphi\wedge K),\cdot\rangle$ in $L^1(\Gamma;\mu_{z\sigma_{\beta\psi}}^{\beta\phi})$ and integrability of $\nabla_\gamma F$ and $\langle\nabla\psi,\cdot\rangle$ w.r.t.~$\mu_{z\sigma_{\beta\psi}}^{\beta\phi}$, we obtain 
\begin{multline*}
\int_{\Gamma} \nabla_\gamma^\Gamma F \cdot h(\langle\varphi\wedge K,\cdot\rangle))\,d\mu_{z\sigma_{\beta\psi}}^{\beta\phi}=\\
-\int_{\Gamma} F\cdot h'(\langle\varphi\wedge K,\cdot\rangle))\langle \nabla(\varphi\wedge K),\cdot\rangle\,d\mu_{z\sigma_{\beta\psi}}^{\beta\phi}+\beta\int_{\Gamma} F \cdot h(\langle\varphi\wedge K,\cdot\rangle))\langle \nabla\psi,\cdot\rangle\,d\mu_{z\sigma_{\beta\psi}}^{\beta\phi}
\end{multline*}
for any $h\in C_b^\infty(\R)$. Letting $K\to\infty$ and using similar arguments we obtain
\begin{equation}\label{eqn:pretinvb}
\int_{\Gamma} \nabla_\gamma^\Gamma F \cdot h(\langle\varphi,\cdot\rangle))\,d\mu_{z\sigma_{\beta\psi}}^{\beta\phi}=-\int_{\Gamma} F\cdot h'(\langle\varphi,\cdot\rangle))\langle \nabla\varphi,\cdot\rangle\,d\mu_{z\sigma_{\beta\psi}}^{\beta\phi}+\beta\int_{\Gamma} F \cdot h(\langle\varphi,\cdot\rangle))\langle \nabla\psi,\cdot\rangle\,d\mu_{z\sigma_{\beta\psi}}^{\beta\phi}
\end{equation}
for $h\in C_b^\infty(\R)$. Now choose a sequence $(h_k)_{k\in\N}\subset C_b^\infty(\R)$ such that $0\leq h_k\uparrow e^{\beta\cdot}$ and $0\leq h_k'\uparrow \beta e^{\beta\cdot}$ as $k\to\infty$. Taking $h=h_k$ in \eqref{eqn:pretinvb} and letting $k\to\infty$, we obtain \eqref{eqn:pretinv} from the monotone convergence theorem (when considering the positive and negative parts of all components of the integrands in \eqref{eqn:pretinvb} separately). For doing so, we only need to verify that $\nabla_\gamma^\Gamma F e^{\langle\beta\varphi,\cdot\rangle}$, $F e^{\langle\beta\varphi,\cdot\rangle} \langle\nabla\varphi,\cdot\rangle$ and $F e^{\langle\beta\varphi,\cdot\rangle} \langle\nabla\psi,\cdot\rangle$ are $\mu_{z\sigma_{\beta\psi}}^{\beta\phi}$-integrable. For the first two expressions this is clear by Lemma \ref{lem}(ii) and since $\varphi\leq \psi$ and $\nabla\varphi\in L^1(\R^d;m)$. For the last one, we compute using Lemma \ref{lem}(ii) and the assumptions on $\varphi$ and $\psi$
\begin{eqnarray*}
\lefteqn{\int_{\Gamma}\big\vert F e^{\langle\beta\varphi,\cdot\rangle} \langle\nabla\psi,\cdot\rangle \big\vert d\mu_{z\sigma_{\beta\psi}}^{\beta\phi}}\\
& &\leq \frac{\Xi_{\psi-\varphi}}{\Xi_\psi}\Vert F\Vert_\infty \int_{\Gamma} \langle\vert\nabla\psi\vert,\cdot\rangle d\mu_{z\sigma_{\beta(\psi-\varphi)}}^{\beta\phi}\leq \xi_{\sigma_{\beta(\phi-\psi)}}\frac{\Xi_{\psi-\varphi}}{\Xi_\psi}\,\Vert F\Vert_\infty \int_{\R^d} \vert\nabla\psi\vert e^{\beta(\varphi-\psi)}\,dx\\
& &\leq \xi_{\sigma_{\beta(\phi-\psi)}} \frac{\Xi_{\psi-\varphi}}{\Xi_\psi}\,\Vert F\Vert_\infty \left(\int_{\R^d\setminus A} \vert \nabla\psi\vert \,dx+e^{\beta\sup_{y\in A}\varphi(y)}\int_{A} \vert \nabla\psi\vert e^{-\beta\psi}\,dx\right)<\infty.
\end{eqnarray*}
This completes the proof of the lemma.
\end{proof}

\begin{remark}\label{rem:notgeneral}
Some comments should be given on the question why Theorem \ref{thm2} is not shown in the generality of Theorem \ref{thm}(ii). After deriving \eqref{eqn:pretinv} one might try an approximation $\varphi_n:=\psi\wedge n$ in order to extend that equation to $\varphi=\psi$, which coincides then with \eqref{eqn:ibp}. However this seems to lead to the necessity of proving that
$$
\int_{\R^d} \vert \nabla e^{\beta\psi\wedge n-\beta\psi}\vert\,dx=\int_{\R^d} 1_{\{\psi\geq n\}} \vert \nabla \psi\vert  e^{\beta n-\beta\psi}\,dx\to 0,
$$
which is wrong in general if $\psi$ is not weakly differentiable. (In contrast, for the third approximation in the proof of Theorem \ref{thm}(ii) we only needed $\int_{\R^d} 1_{\{\psi\geq n\}}\vert \nabla \psi\vert e^{-\beta\psi}\,dx\to 0$.) We avoid this problem by confining ourselves to treating the case where $\psi$ is weakly differentiable except on a very small set.
\end{remark}

\begin{proof}[Proof of Theorem \ref{thm2}]

Necessity is stated in Theorem \ref{thm}(ii). We prove sufficiency: Let $F\in\mathcal FC_b^\infty(C_0^\infty(\R^d),\Gamma)$ and let $v\in\R^d$. We need to show $\int_{\Gamma} F(\gamma+v)-F(\gamma)\,d\mu_{zm}^{\beta\phi}=0$. For $\varepsilon>0$ let $U_\varepsilon$ consist of those points of $\R^d$ which have distance less than $\varepsilon$ from the line $\{sv\,|\,s\in [0,1]\}$ and choose a function $\chi_\varepsilon\in C^\infty(\R^d)$ such that $\chi_\varepsilon=1$ on $U_{\varepsilon}$ and $\chi_\varepsilon=0$ on $\R^d\setminus U_{2\varepsilon}$. Let $g: \R\to [0,1]$ be a smooth function fulfilling $g(0)=1$ and $g(s)=0$ for all $s\in [1,\infty)$. Choose a smooth function $h_\varepsilon: \R^d\to [0,1]$ such that $h_\varepsilon=1$ outside $B_\varepsilon(0)$ and $h_\varepsilon=0$ in $B_{\varepsilon/2}(0)$. Define $\varphi_\varepsilon:= \psi h_\varepsilon+(1-h_\varepsilon)\inf_{y\in\R^d}\psi(y)$. Then $\varphi_\varepsilon$ fulfills the conditions of Lemma \ref{lem2} and we obtain for all $s\in [0,1]$
\begin{multline}\label{eqn:gurk}
\int_{\Gamma} \nabla_\gamma^\Gamma (F g(\langle \chi_\varepsilon,\cdot\rangle))(\gamma+sv) e^{\beta\langle \varphi_\varepsilon-\psi,\gamma\rangle}d\mu_{zm}^{\beta\psi}(\gamma)\\
=\beta \int_\Gamma F(\gamma+sv)\,g(\langle \chi_\varepsilon,\gamma+sv\rangle) \langle \nabla\psi-\nabla\varphi_\varepsilon,\gamma\rangle e^{\beta\langle \varphi_\varepsilon-\psi,\gamma\rangle}\,d\mu_{zm}^{\beta\phi}(\gamma).
\end{multline}
The choice of $\chi_\varepsilon$ and $g$ implies that for any $\gamma\in\Gamma$ fulfilling $\gamma\cap B_{\varepsilon}(0)\neq \emptyset$ it holds $F(\gamma+sv) g(\langle\chi_\varepsilon,\gamma+sv\rangle)=0$ and $\nabla_\gamma^\Gamma(F g(\langle \chi_\varepsilon,\cdot\rangle)(\gamma+sv)=0$, so the integrands in the above equation can only be nonzero for $\gamma\in\Gamma$ fulfilling $\gamma\cap B_\varepsilon(0)=\emptyset$. Since for all such $\gamma$ we have $\varphi_\varepsilon(x)=\psi(x)$ and $\nabla\varphi_\varepsilon(x)=\nabla\psi(x)$ for all $x\in\gamma$, it follows
$$
\int_\Gamma \nabla_\gamma^\Gamma (Fg(\langle\chi_\varepsilon,\cdot\rangle))(\gamma+sv) d\mu_{zm}^{\beta\phi}(\gamma)=0.
$$
Since $\frac{d}{ds} (F(\gamma+sv)g(\langle\chi_\varepsilon,\gamma+sv\rangle))=v\nabla_\gamma^\Gamma (Fg(\langle\chi_\varepsilon,\cdot\rangle))(\gamma+sv)$ for $\gamma\in\Gamma$, $s\in [0,1]$, it follows from the fundamental theorem of calculus that
$$
\int_{\Gamma} F(\gamma+v) g(\langle \chi_\varepsilon,\gamma+v\rangle) d\mu_{zm}^{\beta\phi}(\gamma)=\int_{\Gamma} F(\gamma) g(\langle \chi_\varepsilon,\gamma)\rangle d\mu_{zm}^{\beta\phi}(\gamma).
$$
Letting $\varepsilon\to 0$, we obtain $\chi_\varepsilon\to 0$ Lebesgue-a.e.; here we use that $d\geq 2$. Hence by Lebesgue's theorem and Lemma \ref{lem}(i) it follows
$$
\int_{\Gamma} F(\gamma+v)\,d\mu_{zm}^{\beta\phi}(\gamma)=\int_{\Gamma} F(\gamma)\,d\mu_{zm}^{\beta\phi}(\gamma),
$$
which is what we needed to show.
\end{proof}

\section*{Appendix}

Let us recall the definitions of superstability and lower regularity of a potential and some definitions from Gibbs measure theory. We call a function $\phi: \R^d\to\R\cup\{\infty\}$ a \emph{potential}, if it is measurable and even (i.e.~$\phi(x)=\phi(-x)$ for all $x\in\R^d)$. $\phi$ said to be \emph{superstable}, if there are $a>0$ and $b\geq 0$ such that for any finite configuration $\gamma$ it holds
$$
\sum_{\{x,x'\}\in\gamma}\phi(x-x')\geq a\sum_{r\in\Z^d}\sharp(\gamma\cap Q_r)^2-b\sharp \gamma,
$$
where $Q_r:=\{(x_1,\cdots,x_d)\in\R^d\,|\,r_i-1/2<x_i\leq r_i+1/2, 1\leq i\leq d\}$ for $r=(r_1,\cdots,r_d)\in\Z^d$, and $\sharp M$ denotes the cardinality of a set $M$. It is called \emph{stable}, if the above estimate holds with $a=0$. $\phi$ is called \emph{lower regular}, if there exists a decreasing function $\theta: [0,\infty)\to [0,\infty)$ such that $\int_0^\infty r^{d-1}\theta(r)\,dr<\infty$ and $\phi(x)\geq -\theta(\vert x\vert)$, $x\in\R^d$. These conditions are fulfilled by a wide class of potentials including those of Lennard-Jones type.

By $\Gamma_0$ we denote the set of finite elements of $\Gamma$, and equip it with the trace $\sigma$-field $\mathcal B_0$ of $\mathcal B$ corresponding to the inclusion $\Gamma_0\subset\Gamma$. If $\Lambda\subset\R^d$ is measurable, we set $\Gamma_{\Lambda}:=\{\gamma\in\Gamma\,|\,\gamma\subset\Lambda\}$. It can be considered as a subset of $\Gamma$ or, if $\Lambda$ is relatively compact, as a subset of $\Gamma_0$. Given a $\sigma$-finite measure $\sigma$ on $\R^d$ and an activity parameter $z>0$, one defines on $\Gamma_0$ the Lebesgue-Poisson measure $\lambda_{z\sigma}$ by
$$
\lambda_{z\sigma}(A):=\sum_{n=0}^\infty \frac{z^n}{n!} \int_{(\R^d)^n} 1_A(\{x_1,\cdots,x_n\})d\sigma(x_1)\cdots d\sigma(x_n),\quad A\in\mathcal B_0.
$$
A measure $\mu$ on $(\Gamma,\mathcal B)$ is said to be \emph{tempered} if it is supported on the set
$$
\bigcup_{N\in\N}\bigg\{\gamma\in\Gamma\,\bigg|\,\sum_{r\in[l]}\sharp(\gamma\cap Q_r)^2\leq N^2 (2l+1)^d \mbox{ for all $l\in\N$}\bigg\},
$$ 
where $[l]:=\Z^d\cap [-l,l]^d$. Let $\phi$ be a stable potential, $\beta>0$ and $z>0$. If a tempered measure $\mu$ on $(\Gamma,\mathcal B)$ fulfills the following condition (the Ruelle equation):
\begin{enumerate}
\item[(R)] For any nonnegative $\mathcal B$-measurable $F: \Gamma\to\R$ and all measurable relatively compact $\Lambda\subset\R^d$ it holds
$$
\quad\quad\int_{\Gamma}F\,d\mu=\int_{\Gamma_{\R^d\setminus \Lambda}}\int_{\Gamma_{\Lambda}} F(\gamma\cup\eta)e^{-\beta \sum_{x\in\eta,y\in\gamma}\phi(x-y)-\beta \sum_{\{x,x'\}\subset \eta} \phi(x-x')}\,d\lambda_{z\sigma}(\eta)d\mu(\gamma),
$$
\end{enumerate}
then it is said to be a tempered \emph{grand canonical Gibbs measure} for $\phi$ with intensity measure $\sigma$, inverse temperature $\beta$ and activity $z$. 

\vspace{2ex}\textbf{Acknowledgement}: The authors thank the CCM at the University of Madeira for the hospitality during the Madeira Math Encounters XXXVII in 2009, where this work was initiated. Financial support through FCT, POCTI-219, FEDER, the SFB 701 and DFG projects GR 1809/5-1 and GR 1809/8-1 is gratefully acknowledged.

\bibliography{ANote}

\begin{thebibliography}{CFGK11}

\bibitem[AKR98]{AKR98b}
S.~Albeverio, Yu.G. Kondratiev, and M.~R{\"o}ckner.
\newblock Analysis and geometry on configuration spaces: The {G}ibbsian case.
\newblock {\em J.~Funct.~Anal.}, 157:242--291, 1998.

\bibitem[CFGK11]{CFGK11}
F.~Conrad, T.~Fattler, M.~Grothaus, and Yu.G. Kondratiev.
\newblock An invariance principle for a tagged particle process in continuum
  with singular interactions.
\newblock In preparation, 2011.

\bibitem[FG10]{FG11}
T.~Fattler and M.~Grothaus.
\newblock Tagged particle process in continuum with singular interactions.
\newblock Accepted for publication in Infin.~Dimens.~Anal.~Quantum
  Probab.~Relat.~Top., 2010.

\bibitem[GP87]{GP85}
M.~Z. Guo and G.~Papanicolaou.
\newblock Self-diffusion of interacting brownian particles.
\newblock In {\em Probabilistic methods in mathematical physics}, pages
  113--151. Academic Press, Inc., Boston, Mass., 1987.

\bibitem[KK02]{KK02}
Yu.G. Kondratiev and T.~Kuna.
\newblock Harmonic analysis on configuration space{. I. Gener}al theory.
\newblock {\em Infin. Dimens. Anal. Quantum Probab. Relat. Top.},
  5(2):201--233, 2002.

\bibitem[KK03]{KK03}
Yu.G. Kondratiev and T.~Kuna.
\newblock Correlation functionals for {Gibbs measures and Rue}lle bounds.
\newblock {\em Methods Funct. Anal. Topol.}, 9(1):9--58, 2003.

\bibitem[Pap87]{Pap87}
F.~Papangelou.
\newblock On the absence of phase transition in continuous one-dimensional
  gibbs systems with no hard core.
\newblock {\em Prob.~Theo.~Relat. Fields}, 74:485--496, 1987.

\bibitem[Rue70]{Ru70}
D.~Ruelle.
\newblock Superstable interactions in classical statistical mechanics.
\newblock {\em Commun. Math. Phys.}, 18:127--159, 1970.

\bibitem[Str09]{S09}
S.~Struckmeier.
\newblock {\em Asymptotic {Behavior of some Stochastic Evolutions in
  Conti}nuum}.
\newblock PhD thesis, Bielefeld University, 2009.

\end{thebibliography}
\bibliographystyle{alpha}

\end{document}